\documentclass[12pt, reqno]{amsart}
\date{\today}

\usepackage[english]{babel}
\usepackage{mathrsfs, csquotes}

\usepackage{mathrsfs}               
\usepackage{amssymb}
\usepackage{graphicx, amsmath}
\usepackage{amsfonts}
\usepackage{amssymb}
\usepackage{fancyhdr}
\usepackage{a4wide}
\usepackage{xcolor}

\usepackage{color} 
\definecolor{bleu_sombre}{rgb}{0,0,0.6} 
\definecolor{bs}{rgb}{0,0,0.6}  
\definecolor{rouge_sombre}{rgb}{0.8,0,0}
\definecolor{rs}{rgb}{0.8,0,0}
\definecolor{vert_sombre}{rgb}{0,0.6,0}
\definecolor{vs}{rgb}{0,0.6,0}
\usepackage[plainpages=false,colorlinks,linkcolor=bleu_sombre,
citecolor=rouge_sombre,urlcolor=vert_sombre,breaklinks]{hyperref}

\usepackage{dsfont}
\newcommand{\un}{\mathds{1}}

\usepackage[normalem]{ulem}

\newcommand{\C}{\mathbb{C}} 
\newcommand{\Z}{\mathbb{Z}} 
\newcommand{\R}{\mathbb{R}} 
\newcommand{\tr}{\mathrm{tr}}

\newtheorem{proposition}{Proposition}[section]
\newtheorem{corollary}[proposition]{Corollary}
\newtheorem{theorem}[proposition]{Theorem}
\newtheorem{lemma}[proposition]{Lemma}

\newtheorem{hypothesis}{Hypothesis}

\newtheorem*{example*}{Examples}

\theoremstyle{remark}
\newtheorem{remark}{Remark}

\newcommand{\sps}[2]{\langle #1,#2 \rangle} 

\newcommand{\rd}{d}

\newcommand{\pauli}{\boldsymbol{\sigma}}

\newcommand{\bx}{\mathbf{x}}

\newcommand{\cb}{}

 \renewcommand{\ge}{\geqslant}
      
\newcommand{\oc}{\tilde{\Omega}} 
\newcommand{\og}{{\Omega'}} 
\newcommand{\dom}{\mathcal{D}}%
\newcommand{\normal}{\mathbf{n}}
\newcommand{\curv}{\kappa}
\newcommand{\hmv}{H_m^V}
\newcommand{\hov}{H_\Omega^V}
\newcommand{\rmv}{R_m^V}
\newcommand{\rov}{R_\Omega^V}

\title[Resolvent convergence to Dirac operators defined on
domains]{Resolvent convergence to Dirac operators on planar  domains}
\author{Jean-Marie Barbaroux}
\address{Jean-Marie Barbaroux\\
 Aix Marseille Univ, Universit\'e de Toulon\\ CNRS, CPT, Marseille, France.}
\email{barbarou@univ-tln.fr}
\author{Horia Cornean}
\address{Horia Cornean\\
Department of Mathematical Sciences, Aalborg University\\
Fredrik Bajers Vej 7G, 9220 Aalborg \O, Denmark.}
\email{cornean@math.aau.dk}
\author{Lo\"{i}c Le Treust}
\address{Lo\"{i}c Le Treust\\
Aix Marseille Univ, CNRS, Centrale Marseille, I2M\\
 Marseille, France.}
\email{loic.le-treust@univ-amu.fr}
\author{Edgardo Stockmeyer}
\address{ Edgardo Stockmeyer\\ Instituto de F\'\i sica\\
Pontificia Universidad Cat\'olica de Chile\\
Vicu\~na Mackenna 4860\\
 Santiago 7820436, Chile.}
\email{stock@fis.puc.cl}

\begin{document}
\begin{abstract}
  Consider a Dirac operator defined on the whole plane with a mass
  term of size $m$ supported outside a domain $\Omega$. We give a
  simple proof for the norm resolvent convergence, as $m$ goes to
  infinity, of this operator to a Dirac operator defined on $\Omega$
  with infinite mass boundary conditions.  The result is valid
  for bounded and unbounded domains and gives estimates on the speed
  of convergence. Moreover, the method easily extends when adding
 external matrix-valued potentials.
\end{abstract}

\maketitle
\section{Introduction}
On the plane $\R^2$ consider a smooth open set $\Omega$, not
necessarily bounded, with boundary $\partial \Omega$. For $m>0$, let
$H_m$ be the standard two-dimensional Dirac operator, defined
on the whole plane, with a mass term being zero on $\Omega$ and $m$ on
$\oc:=\R^2\setminus \Omega$. Berry and Mondragon \cite{Berry}
realised that in the limit when $m\to \infty$ the operator $H_m$
converges -- in some sense -- to an operator $H_\Omega$, defined only
on $\Omega$, with certain prescribed conditions at the boundary (which
were later known as infinite mass boundary conditions). Recently in
\cite{Stockmeyerwugalter2016} it was shown, assuming $\Omega$ to be
bounded, that the spectral projections of $H_m$ converge to those of
$H_\Omega$ as $m\to\infty$. The strategy used in
\cite{Stockmeyerwugalter2016} was based on the mini-max principle,
which is only appropriate for bounded domains. In this work we extend
the results from \cite{Stockmeyerwugalter2016} in several directions.

In recent years, Dirac operators have regain much attention since
they provide an effective description for the dynamics of charge
carriers in graphene \cite{castro2009electronic} and in other novel
materials that exhibit interesting physical  properties (see e.g.,
\cite{armitage2018weyl}).  In the physics literature,  this type of spatially localised mass terms
in the Dirac equation are used to model obstacles for the effective
particles \cite{AkhmerovBeenakker,science,weinstein18}. In addition,
the limiting operator
$H_\Omega$ is used to model particles
confined to $\Omega$, as it is the case for quantum dots
\cite{de2009electron} or for a periodically perforated sheet of
graphene \cite{brun2014electronic,BPP}. It is therefore desirable to
estimate the speed of convergence of $H_m$ to $H_\Omega$. 

In this article we allow $\Omega$ to be unbounded and show that, as
$m\to\infty$, the norm resolvent difference between $H_m$ and
$H_\Omega$ is of order $1/\sqrt{m}$, uniformly in the spectral
parameter (see Theorem \ref{main-thm}). Our proof is based on a novel resolvent
identity (see Lemma \ref{comparison}), which we believe is interesting in its own right.
The simplicity of the proof allows naturally the inclusion of
matrix-valued external potentials (see Theorem \ref{main-thm-2}).


\subsection{Setting and main result}
Let $\Omega\subset \R^2$ be an open set satisfying the following.
\begin{hypothesis}\label{hyp0}
$\ $

\begin{enumerate}
	\item Its boundary $\partial \Omega$ is $C^{2,1}$,
	\item the curvature $\kappa_\Omega$ of $\partial\Omega$ is in $L^\infty(\partial \Omega)$,
	\item there exist global tubular coordinates for a neighborhood of $\partial \Omega$.
Namely, there exist $\delta_0>0$ and $C_0>0$ such that the mapping 
\[\begin{split}
	\Psi : &]-\delta_0,\delta_0[\times \partial \Omega \to \mathbb{R}^2\\
	&(t,s)\mapsto s-tn(s)\,,
\end{split}\]
satisfies
\begin{enumerate}
	\item $\Psi$ is a $C^1$-diffeomorphism on its range,
	\item $\partial \Omega = \Psi(\{0\}\times \partial \Omega)$,
	\item $\Psi(]0,\delta_0[\times \partial \Omega) = \Omega\cap {\rm range}(\Psi)$,
	\item $C_0>|\det J_\Psi(s,t)|>C_0^{-1}$, for all $(t,s)\in ]-\delta_0,\delta_0[\times \partial \Omega$ .
\end{enumerate}
\end{enumerate}
\end{hypothesis}
%
%
%

This includes for instance the following cases:

\begin{itemize}
    \item[(i)] {\rm \bf Bounded boundary case}. The set $\Omega$ or $\tilde{\Omega}$ 
    is a bounded set, with $C^{2,1}$ boundary. 
	\item[(ii)] {\rm \bf Deformed half-space}. The case in which $\Omega$ is a connected unbounded set given by $\{ \bx=(x_1, x_2) \in \R^2\ :\ x_2\in\R \mbox{ and } -\infty < x_1 < f(x_2) \}$ for some $f\in C^{2,1}(\R)$. 
	\item[(iii)] {\rm \bf Periodic holes}.
Another typical situation in which we are interested in is the case where, for fixed $L>0$, $\oc =  \bigcup_{\gamma\in L \Z^2} B_r(\gamma)$ consists of periodically distributed balls of radius $r< L$. 
\end{itemize}

We work on the spaces of measurable square integrable functions
$L^2(\Omega,\C^2)$, $L^2(\oc,\C^2),$ $ L^2(\partial\Omega,\C^2)$, and
$L^2(\R^2,\C^2)$, with scalar products
$\sps{\cdot}{\cdot}_{\Omega}, \sps{\cdot}{\cdot}_{\oc},
\sps{\cdot}{\cdot}_{\partial\Omega} $,
and we set $\sps{\cdot}{\cdot}\equiv \sps{\cdot}{\cdot}_{\R^2}$. We
indicate the set of integration in the same way for the corresponding
norms $\|\cdot\|^2=\sps{\cdot}{\cdot}$. The Euclidian scalar product
on $\C^2$ is denoted by $(\cdot,\cdot)$. Thus, we have, for instance,
for $f,g\in L^2(\Omega,\C^2)$
\begin{align*}
  \sps{f}{g}_{\Omega}=\int_\Omega (f(\bx),g(\bx)) \rd\bx\,.
\end{align*}
Throughout this paper we use the
identification
\begin{equation}\label{splitting}
\begin{split}
L^2(\R^2,\C^2) &\cong L^2(\Omega,\C^2)\oplus L^2(\oc,\C^2)\\
\varphi &\mapsto \varphi\un_\Omega\oplus \varphi\un_{\oc}\,,
\end{split}
\end{equation}
where $\un$ denotes the indicator function on the corresponding
set.

For $\og\in\{\Omega, \oc\}$ we denote by $H^1(\og,\C^2)$ the Sobolev
space of functions in $L^2(\og,\C^2)$ whose first-order distributional
derivatives belong to $L^2(\og,\C^2)$. If it is clear from the
context we may drop the reference to the spinor space $\C^2$ and
simply write $L^2(\og),$ $H^1(\og)$, etc. Due to the regularity of
the boundary, there exist a
continuous linear operator, the trace operator,
\begin{align*}
  \tr_{\og}: H^1(\og,\C^2)\to L^2(\partial\Omega,\C^2),
\end{align*}
such that $\tr_{\og}\varphi=\varphi$ on $\partial\Omega$ for all 
$\varphi\in H^1(\og,\C^2)$ continuous on 
$\overline{\og}$ (see \cite[Theorem 15.23]{Leoni} for a general 
statement valid in the case of unbounded boundaries).

Moreover, for any $\varepsilon>0$ there exists a
constant $C_\varepsilon$ such that for all $\varphi\in H^1(\og,\C^2)$
\begin{align}\label{trace-ineq}
  \|\tr_{\og} \varphi\|^2_{\partial \Omega}\leqslant \varepsilon \|\nabla
  \varphi\|_{\og}^2+
  C_\varepsilon  \|\varphi\|_{\og}^2.
\end{align}
\begin{remark}\label{rem:trace-inequality}
The above trace inequality is usually stated with $\varepsilon=1$. To obtain the desired result  \eqref{trace-ineq}, following the standard proof, it suffices to apply $2ab\leqslant \epsilon a^2 + {b^2}/{\epsilon}$ (e.g. in \cite[proof of theorem~1, \S~5.5, Eq.(1),  p272]{Evans}). 
\end{remark}

We define the differential expression associated to  the standard  massless   Dirac 
operator on the plane as
\begin{align*}
  T:=\frac{1}{i} {\boldsymbol \sigma}\cdot \nabla=\left ( \begin{array}{ll}
                                                            \qquad 0 & -i\partial_1-\partial_2\\
                                                            -i\partial_1+\partial_2
                                                                     &
                                                                       \qquad  0 
\end{array} \right ),
        \end{align*}
where $ {\boldsymbol \sigma}=(\sigma_1,\sigma_2)$ and the Pauli
matrices  are given by
\begin{align*}
  \sigma_1
  =\left(
\begin{array}{cc}
 0&1\\
 1&0
\end{array}
 \right),\quad
\sigma_2=\left(
    \begin{array}{cc}
0&-i\\
i&0
\end{array}
   \right),\quad 
   \sigma_3
   =\left(
\begin{array}{cc}
 1&0\\
 0&-1
\end{array}
 \right).
\end{align*}
They satisfy the
anticommutation relation 
\begin{equation}
  \label{eq:6}
  \{\sigma_j,\sigma_k\}:=\sigma_j\sigma_k+
\sigma_k\sigma_j=2\delta_{jk},\qquad j,k\in\{1,2,3\}.
\end{equation} 
\color{blue}
\color{black}
Let us now introduce the main operators of this work. For $m>0$ we
denote by $H_m$ the self-adjoint operator in $L^2(\R^2,\C^2)$ acting
as 
\begin{align*}
  H_m=T   + m\un_{\oc} \sigma_3
\end{align*}
with operator domain $\dom(H_m)=H^1(\R^2,\C^2)$. For an element of its
resolvent set $z\in\varrho(H_m)$ we write the resolvent of $H_m$ as
\begin{align}
  \label{eq:1}
  R_m(z):=(H_m-z)^{-1}.
\end{align}
We are interested in the limit of $R_m$ as $m\to\infty$; the limit
operator acts non-trivially in $L^2(\Omega,\C^2)$. In order to describe its boundary
conditions we set, for $\og\in\{\Omega,\oc\}$,
$\normal_{\og}:\partial\Omega \to \boldsymbol{S}^1$ to be the normal
vector to ${\og}$ pointing outwards. For any $\bx\in \partial\Omega$
we define the {\it boundary matrix} as
\begin{align*}
  B_{\og}(\bx):=-i\sigma_3 \pauli\cdot\normal_{\og}(\bx),
\end{align*}
which is self-adjoint with eigenvalues $\pm 1$.  For the case
$\og=\Omega$ we write the corresponding eigenprojections as
$ P_{\pm}:=(1\pm B_{\Omega})/2$.  Due to the anticommutator relations
\eqref{eq:6} the following intertwining relation holds
\begin{equation}
  \label{eq:5}
  \sigma_3P_-=P_+\sigma_3.
\end{equation}
We define the self-adjoint operator (see e.g. \cite{Benguria2017})
\begin{equation}
\label{eq:3}
H_\Omega\varphi= T  \varphi, \qquad \varphi\in\dom_{\Omega}:=\{\varphi\in H^1(\Omega,\C^2)\,:\,
P_-
\tr_{\Omega} \varphi=0\}.
\end{equation}
In the physics literature conditions as $P_-
\tr_{\Omega} \varphi=0$ or $P_+
\tr_{\Omega} \varphi=0$ are referred to as infinite-mass boundary 
conditions \cite{AkhmerovBeenakker}. For further recent study of $H_\Omega$ from the mathematical point of view see e.g. \cite{BFSV,treust2017spectral,Borrelli}.

Note that in dimension $3$, this infinite mass limit is strongly related to the MIT bag model for quarks confinement. This model has been derived first by Bogolioubov \cite{bogoliubov1987} in 1968 and extended by physicists from MIT into a shape optimization problem  \cite{MIT101974} in 1974. We also refer to \cite{arrizabalaga:hal-01863065,arrizabalaga:hal-01540149,arrizabalaga:hal-01343717} for a study of these problems from the mathematical viewpoint and other physical references on the subject. 


For $z\in\varrho(H_\Omega)$ we write 
\begin{align}
  \label{eq:1.1}
  R_\Omega(z):=(H_\Omega-z)^{-1}.
\end{align}
\begin{theorem}[Resolvent convergence - free case]\label{main-thm}
Let $K\subset \varrho(H_\Omega)$ be a compact set. Then, there exists
$m_0>1$ such that for all $m>m_0$: $K\subset\varrho(H_m)$ and there is a
constant $C_{K}\equiv C<\infty$ such that
\begin{equation}
  \label{eq:7}
  \sup_{z\in K}\|R_m(z)-R_\Omega(z)\oplus 0\|\leqslant \frac{C}{\sqrt{m}}\,.
\end{equation}
Here the orthogonal sum is in the sense of the splitting
\eqref{splitting}.
\end{theorem}
%
%
%
\begin{remark}\label{rmk2}
  Notice that Theorem \ref{main-thm} implies that the spectrum of
  $H_m$ converges to that of $H_\Omega$ (see e.g., \cite[Proposition
  10.2.4]{Oliveira:2008kq}). 
  Consequently, for
  $a,b\in \varrho(H_\Omega)$ with $a<b$, we have
 $$
 \||\un_{(a,b)}(H_m)-\un_{(a,b)}(H_\Omega)\oplus 0\|=\mathcal{O}(m^{-1/2}),\qquad m\gg
  1,
 $$
where $\un_{(a,b)}$ is the indicator function on $(a,b)$ and the
spectral projections are defined through the functional
calculus. This result is obtained by applying the estimate \eqref{eq:7} to the Riesz'
  Formula $\un_{(a,b)}(H_m)-\un_{(a,b)}(H_\Omega)\oplus 0 
  = \frac{1}{2 i \pi}\int_\gamma R_m(z)-R_\Omega(z)\oplus 0\, d z$, where $\gamma$ is a 
  contour  around $(a,b)$ with distance $\nu>0$ to $(a,b)$.

In particular,  if $\lambda$ is an isolated eigenvalue of $H_\Omega$, there exists a 
sequence $(\lambda_m)$ such that for all $m$ large enough $\lambda_m$ 
is an eigenvalue 
of $H_m$ and $\lambda - \lambda_m = \mathcal{O}(m^{-1/2}),\ m\gg 1$. 
\end{remark}
\begin{remark}\label{rmk3}
  It is easy to check that $H_{-m}$ convergences, as $m\to \infty$, to
  an operator defined as in \eqref{eq:3} but with the boundary
  condition replaced by $P_+ \tr_{\Omega} \varphi=0$ (see \cite[Remark
  2]{Stockmeyerwugalter2016}).
\end{remark}
In \cite{Stockmeyerwugalter2016} it was shown that the spectral
projections of $H_m$ convergence to those of $H_\Omega$ in operator
norm.  In their proof the authors used a combination of the mini-max
principle (which is not suitable for unbounded domains) with certain
apriori lower bound on the quadratic form
$\varphi\mapsto \|H_m\varphi\|^2$. In this work we use the same lower bound
(see \cite[Lemma~4]{Stockmeyerwugalter2016} and
Proposition~\ref{prop:sa} below) together with certain novel resolvent
identity stated in Lemma~\ref{comparison} below. The main point here is that we compare, through their
resolvents, the operator $H_m$ with $H_{\Omega}\oplus H_{\oc}^{(m)}$,
where $H_{\oc}^{(m)}=T+m\sigma_3$ is an auxiliar operator having
domain  $H^1(\oc)$ with infinite-mass
boundary conditions (see \eqref{eq:4.1} below).
We conclude the proof
by estimating the corresponding resolvents and resolvent traces.
The proof of Theorem~\ref{main-thm} is given in Section~\ref{pr}.
%
%
%
%

\medskip
The next theorem is an extension of the main theorem to the case of 
Dirac operators with matrix-valued potentials. 
\begin{hypothesis}\label{hyp1}
Let $V$ be a symmetric matrix-valued potential. We assume: 

(i) There exists $a\in [0,1)$ and $b\geqslant 0$ such that for all $\varphi\in H^1(\Omega,\C^2)$
$$
 \| V\varphi \|^2_\Omega \leqslant a \|\nabla \varphi\|_\Omega^2 + b \|\varphi\|_\Omega^2 \, .
$$

(ii) The potential $V$ is bounded in 
$\oc$: $\|V\|_{L^\infty(\oc)} <\infty$.
\end{hypothesis}

We thus consider the following operators
$$
 \hmv = H_m + V 
$$
and for $z\in\rho(\hmv)$ we denote
$$
 \rmv(z) =  (\hmv - z)^{-1}\, .
$$ 
We also define
\begin{equation}\nonumber
\hov = H_\Omega +V \quad\mbox{on}\quad \{\varphi\in H^1(\Omega,\C^2)\,:\,
P_-
\tr_{\Omega} \varphi=0\}
\end{equation}
and $\rov(z) =  (\hov - z)^{-1}$ for $z\in\rho(\hov)$. 
%
%
%
\begin{proposition}\label{prop:sa}
Assume Hypothesis~\ref{hyp1} holds. Then $\hmv$ is self-adjoint on $\dom(\hmv) = H^1(\R^2, \C^2)$ and $\hov$ is self-adjoint on $\dom(\hov) = \{\varphi\in H^1(\Omega,\C^2)\,:\,P_- \tr_{\Omega} \varphi=0\}$.
\end{proposition}
%
%
%
\begin{theorem}[Resolvent convergence - potential case]\label{main-thm-2}
Assume Hypothesis~\ref{hyp1} holds. Let $K\subset \varrho(H_\Omega)$ be a compact set. Then, there exists
$m_0>1$ such that for all $m>m_0$: $K\subset\varrho(H_m)$ and there is a
constant $C_{K}^V\equiv C<\infty$ such that
\begin{equation}
  \sup_{z\in K}\|\rmv(z)-\rov(z)\oplus 0\|\leqslant \frac{C}{\sqrt{m}}\,.
\end{equation}
\end{theorem}
The proofs of the above two results are postponed to Section~\ref{S4}.
%
%
%
\section{Proof of Theorem~\ref{main-thm}}\label{pr}
For $m>0$ we define the auxiliar self-adjoint operator acting in
$L^2(\oc,\C^2)$ as
\begin{equation}
\label{eq:4.1}
H_{\oc}^{(m)}\varphi= T\varphi+m\sigma_3\varphi, \qquad\varphi\in\dom_{\oc}:=\{\varphi\in H^1(\oc,\C^2)\,:\,
P_+\tr_{\oc} \varphi= 0\}.
\end{equation}
For $z\in\varrho(H_{\oc}^{(m)})$ we write its resolvent  
\begin{align}
  \label{eq:1.2}
  R_{\oc}^{(m)}(z):=(H_{\oc}^{(m)}-z)^{-1}.
\end{align}
\begin{remark}\label{B-bc}
  Notice that, since $ B_{\Omega}=-B_{\oc}$, the boundary conditions
  stated in \eqref{eq:3} and \eqref{eq:4.1} can  both be written in
  terms of the corresponding boundary matrix. More precisely, for 
$\og\in\{\Omega,\oc\}$, we have  
  \begin{align*}
    \varphi\in \mathcal{D}_{\og}\,\,\,\Leftrightarrow \,\,\, \varphi\in
    H^1(\og,\C^2)
\quad \mbox{and}\quad B_{\og}\tr_{\og} \varphi= \tr_{\og} \varphi\,.
  \end{align*}
\end{remark}
\begin{lemma}[Resolvent identity]
 \label{comparison}
 Consider the resolvents \eqref{eq:1},
 \eqref{eq:1.1}, and \eqref{eq:1.2}, evaluated at $z\in \varrho(H_m)\cap \varrho(H_{\Omega}) \cap
 \varrho(H_{\oc}^{(m)})$.
 Then, for any $f,g \in L^2(\R^2,\C^2)$ with $g=h\oplus \tilde{h}\in
 L^2(\Omega,\C^2)\oplus L^2(\oc,\C^2) $, we have
\begin{align*}
  \sps{f}{R_m g}=\sps{f}{R_{\Omega}\oplus R_{\oc}^{(m)} g}+
\sps{P_-\tr_{\Omega} R_{m}^* f}{\sigma_3 \tr_{\Omega}
  R_{\Omega} h}_{\partial\Omega}
+
\sps{P_+\tr_{\oc} R_{m}^* f}{\sigma_3 \tr_{\oc} R_{\oc}^{(m)} \tilde{h}}_{\partial\Omega}\,.
\end{align*}
\end{lemma}
\begin{proof}
Let $\psi\in H^1(\R^2,\C^2)$ and $\varphi=\phi\oplus\tilde{\phi}\in
\dom_{\Omega}\oplus\dom_{\oc} $. Since the mass term of $H_m$ is
supported on $\oc$ we have that
\begin{align*}
  \sps{H_m\psi}{\varphi}
=\sps{T\psi}{\phi}_{\Omega}+\sps{H_m\psi}{\tilde{\phi}}_{\oc}
=\sps{T\psi}{\phi}_{\Omega}+\sps{T\psi}{\tilde{\phi}}_{\oc}+\sps{\psi}{m\sigma_3\tilde{\phi}}_{\oc}\,.
\end{align*}
Using partial integration (see Lemma \ref{part-int} below) we get for any $z\in \C$
\begin{align*}
  &\sps{(H_m-\overline{z})\psi}{\varphi}=\sps{\psi}{(H_{\Omega}\oplus H_{\oc}^{(m)}-z) \varphi}-\sps{\tr_{\Omega}
  \psi}{\sigma_3 B_{\Omega} \,\tr_{\Omega} \phi}_{\partial\Omega}
-\sps{\tr_{\oc}
  \psi}{\sigma_3 B_{\oc} \,\tr_{\oc} \tilde{\phi}}_{\partial\Omega}\\
&\quad\quad\quad=\sps{\psi}{(H_{\Omega}-z)\phi}+\sps{\psi}{(H_{\oc}^{(m)}-z) \tilde\phi}-\sps{\tr_{\Omega}
  \psi}{\sigma_3 \,\tr_{\Omega} \phi}_{\partial\Omega}
-\sps{\tr_{\oc}
  \psi}{\sigma_3 \,\tr_{\oc} \tilde{\phi}}_{\partial\Omega}\,,
\end{align*}
where in the last step we used the boundary conditions for $\phi$ and
$\tilde{\phi}$ (see Remark~\ref{B-bc}). Since $\tr_{\Omega}
\phi=P_+ \tr_{\Omega} \phi$ and $\tr_{\oc} \tilde{\phi}=P_-\tr_{\oc}
\tilde{\phi}$ we may write, in view of \eqref{eq:5}, the two boundary
terms above as $ \sps{P_-\tr_{\Omega}
  \psi}{\sigma_3 \,\tr_{\Omega} \phi}_{\partial\Omega}$ and $\sps{P_+\tr_{\oc}
  \psi}{\sigma_3 \,\tr_{\oc} \tilde{\phi}}_{\partial\Omega}$, respectively.

Let $f,g,h,\tilde{h}$, and $z$ be given as in the statement of the lemma.  The
claim follows from the above computation since it holds for
$\psi:=R^*_m f=R_m(\overline{z})f$, $\phi:=R_{\Omega}(z) h$, and
$\tilde{\phi}:=R_{\oc}^{(m)} (z)\tilde{h}$.
\end{proof}
\begin{proof}[Proof of Theorem \ref{main-thm}]
Let $f,g \in L^2(\R^2,\C^2)$ with $g=h\oplus \tilde{h}\in
 L^2(\Omega,\C^2)\oplus L^2(\oc,\C^2) $. 
For fixed constants $\mu_0>0$ and $\rho\geqslant 0$ define the set 
\begin{align}\label{set}
S_{\mu_0,\rho}:=\{\xi\in\C \,:\,|{\rm Im}(\xi)|\geqslant \mu_0 
\quad
\mbox{and}\quad 
|{\rm Re}(\xi)|\leqslant \rho\}\,.
\end{align}
According to Lemma~\ref{nc} below it suffices to
show, for some $C'<\infty$ and $m'<\infty$,
\begin{align*}
 \sup_{\xi\in S_{\mu_0,\rho}} \big| \sps{f}{(R_m(\xi)-R_\Omega(\xi)\oplus 0)\,g}\big|\leqslant\frac{C'}{\sqrt{m}} \|f\|\|g\|, \qquad \mbox{for } m\geqslant m' .
\end{align*}

From the resolvent identity in Lemma 
 \ref{comparison} we get, using Cauchy-Schwarz Inequality, 
 \begin{align*}
   \big|\sps{f}{ (R_m(\xi) -R_{\Omega}(\xi)\oplus 0) g}\big|&\leqslant
\|R_{\oc}^{(m)}(\xi) \tilde{h}\| \|f\|+\|P_-\tr_{\Omega} R_{m}(\bar\xi)
f\|_{\partial\Omega}\,
\|\tr_{\Omega}  R_{\Omega}(\xi) h  \|_{\partial\Omega}\\
&\quad
+\|\tr_{\oc} R_{m}(\bar\xi)
f\|_{\partial\Omega}\, \|\tr_{\oc} R_{\oc}^{(m)}(\xi) \tilde{h} \|_{\partial\Omega},
 \end{align*}
for any $\xi\in S_{\mu_0,\rho}$. The theorem now follows from lemmas
\ref{r-lemma1}, \ref{r-lemma2}, and Corollary \ref{r-lemma3} from
Section~\ref{rt} below.
\end{proof}
%
%
%
The next is an elementary result from partial integration.
\begin{lemma}\label{part-int}
Let $\og\in\{\Omega,\oc\}$. For any $\phi,\psi\in H^1(\og,\C^2)$ we
have
\begin{align}
  \label{eq:9}
  \sps{T\phi}{\psi}_{\og}=\sps{\phi}{T\psi}_{\og} -\sps{\tr_{\og}
  \phi}{\sigma_3 B_{\og} \,\tr_{\og} \psi}_{\partial\Omega}\,.
\end{align}
\end{lemma}
\begin{proof}
Using divergence theorem for the Sobolev functions  $\phi,\psi\in
H^1(\og,\C^2)$ we compute
\begin{align*}
  \sps{T\phi}{\psi}_{\og}&=\int_{\og} (-i \pauli\cdot\nabla\phi (\bx),
  \psi(\bx)) d\bx= \sps{\phi}{T\psi}_{\og}+\int_{\og} \nabla\cdot (\phi (\bx), i\pauli
  \psi(\bx)) d\bx\\
&=\sps{\phi}{T\psi}_{\og}+\int_{\partial\Omega}(\tr_{\og}\phi (\bx),
  i(\pauli\cdot \normal_{\og})\tr_{\og}\psi(\bx)) d\bx .
\end{align*}
The lemma follows now since $ i\pauli\cdot \normal_{\og}=-\sigma_3 B_{\og}$.
\end{proof}
\section{Resolvent estimates}
\label{rt}
Let us start by introducing some additional geometrical objects for $\og\in\{\Omega,\oc\}$.  Given ${\normal}_{\og}=(n_1,n_2)$ we
define the tangent vector ${\bf t}_{\og}:=(-n_2,n_1)$. Normal and
tangent vectors are related to the curvature $\kappa_{\og}$ through
the formula
\begin{align}
  \label{curvature}
({\bf t}_{\og}\cdot \nabla)\, {\bf t}_{\og}=-\curv_{\og} \,\normal_{\og}\,.
\end{align}
Due to Hypothesis~\ref{hyp0} on $\partial\Omega$ we have
$\|\kappa_{\og}\|_\infty\equiv \|\kappa\|_\infty<0$. The following
lemma is well known. We give its proof for completeness at the end of
this section.
\begin{lemma}\label{energysquared}
Let $\og\in\{\Omega,\oc\}$. Then, for any $\varphi\in
\mathcal{D}_{\og}$, we have
\begin{align}
  \label{eq:14}
  \|T \varphi\|_{\og}^2=\|\nabla\varphi\|^2_{\og}+\tfrac{1}{2}  
\sps{\tr_{\og}\varphi}{\curv_{\og} \, \tr_{\og} \varphi}_{\partial\Omega}\,.
\end{align}
\end{lemma}
Let $\mu_0\in (0,\infty)$  and $\rho\in[0,\infty)$ be fixed constants.  In this section we
will evaluate the resolvents on the set $S_{\mu_0,\rho}$ defined in \eqref{set}.
\begin{lemma}\label{r-lemma1}
There exist a constant $c<\infty$ such that for all $\xi\in S_{\mu_0,\rho}$ 
and all $h\in L^2(\Omega,\C^2)$
\begin{align}
  \label{eq:13a}
  &\|R_{\Omega}(\xi) h\|_{\Omega}\leqslant \frac{1}{\mu_0} \|h\|_{\Omega}\, ,\\ \label{eq:13.1}
&\|\tr_{\Omega} R_{\Omega}(\xi) h\|_{\partial\Omega}\leqslant c\|h\|_{\Omega}\,.
\end{align}
\end{lemma}
\begin{proof}
  The first inequality above is straightforward. In order to show
  \eqref{eq:13.1} observe first that,  for all
$\varphi\in\mathcal{D}_{\Omega}$ and $\xi=\eta+i\mu\in S_{\mu_0,\rho}$,
  \begin{align}\label{xi-bound}
    \|(H_\Omega-\xi)\varphi\|^2=
    \|(T-\eta)\varphi\|_{\Omega}^2+\mu^2
    \|\varphi\|_{\Omega}^2\geqslant \tfrac{1}{2}
    \|T\varphi\|_{\Omega}^2+(\mu_0^2-2\rho^2) 
    \|\varphi\|_{\Omega}^2\,.
  \end{align}
Recall that by the Trace Inequality \eqref{trace-ineq}
  for every $\varepsilon>0$ there is a $C_\varepsilon <\infty$ such
  that
  \begin{align}\label{ti}
    \|\nabla \varphi\|^2\geqslant \tfrac{1}{\varepsilon}
    \|\tr_{\Omega}\varphi\|^2_{\partial\Omega}
-\tfrac{C_\varepsilon}{\varepsilon} \|\varphi\|_{\Omega}^2\,,\quad
    \varphi \in H^1(\Omega).
  \end{align}
Thus, using the previous two inequalities and estimating equation \eqref{eq:14} we get, for all
$\varphi\in\mathcal{D}_{\Omega}$ and $\xi\in S_{\mu_0,\rho}$,
\begin{align}\nonumber
  2\|(H_\Omega-\xi)\varphi\|^2&\geqslant
  \big(\tfrac{1}{\varepsilon}-\tfrac{\|\curv\|_\infty}{2}\big)
  \|\tr_{\Omega}\varphi\|^2_{\partial\Omega}+\big(
  2\mu_0^2-4\rho^2-\tfrac{C_\varepsilon}{\varepsilon}\big)
 \|\varphi\|_{\Omega}^2\\\label{e2}
&\geqslant   \|\tr_{\Omega}\varphi\|^2_{\partial\Omega}-\tilde{C}_\varepsilon \|\varphi\|_{\Omega}^2\,,
\end{align}
where in the last step we choose
$\varepsilon^{-1}=1+\|\curv\|_\infty/2$ and set
$\tilde{C}_\varepsilon:=|2\mu_0^2-4\rho^2-C_\varepsilon/\varepsilon|$.  Let 
$h\in L^2(\Omega,\C^2)$. Replacing $\varphi$ by $R_\Omega(\xi)h\in \mathcal{D}_{\Omega}$
in \eqref{e2} we get that
\begin{align*}
  \|\tr_{\Omega}R_\Omega(\xi)h\|^2_{\partial\Omega}\leqslant
  2\|h\|_{\Omega}^2+
\tilde{C}_\varepsilon \|R_\Omega(\xi)h\|^2_{\Omega}\leqslant
  (2+\tilde{C}_\varepsilon/\mu_{0}^2) 
 \|h\|^2_{\Omega}\,.
\end{align*}
This finishes the proof of the lemma.
\end{proof}
%
%
%
\begin{lemma}\label{r-lemma2}
For all $\xi\in
S_{\mu_0,\rho}$,  $h\in L^2(\oc,\C^2)$, and all
$m>2\sqrt{|\mu_0^2-2\rho^2|}$ the following inequalities hold
\begin{align}
  \label{eq:13.2}
  &\|R_{\oc}^{(m)}(\xi) h\|_{\oc}\leqslant \tfrac{2}{m} \|h\|_{\oc}\, ,\\\label{eq:13.3}
&\|\tr_{\oc} R_{\oc}^{(m)}(\xi) h\|^2_{\partial\Omega}\leqslant 
{\tfrac{2}{m}}\|h\|_{\oc}^2\,.
\end{align}
\end{lemma}
\begin{proof}
First observe that, for $\varphi\in \mathcal{D}_{\oc}$,
\begin{equation}
\label{tb}
\begin{split}
  \|H_{\oc}^{(m)}\varphi\|^2_{\oc}&=
\|T\varphi\|^2_{\oc}+m^2\|\varphi\|^2_{\oc}+2\,{\rm Re}
\sps{T\varphi}{m\sigma_3\varphi}_{\oc}\\
&=\|T\varphi\|^2_{\oc}+m^2\|\varphi\|^2_{\oc}+
m\|\tr_{\oc}\varphi\|^2_{\partial\Omega}\,.
\end{split}
\end{equation}
where in the last step we used the boundary conditions (see Remark
\ref{B-bc}) and  that Green's identity gives, for $\varphi\in H^1(\oc)$, 
\begin{align*}
2{\rm Re}
\sps{T\varphi}{\sigma_3\varphi}_{\oc}=i \big(\sps{\pauli\cdot\nabla
  \varphi}{\sigma_3\varphi}_{\oc}+\sps{\varphi}{\pauli\cdot\nabla
  \,\sigma_3\varphi}_{\oc} \big)=\sps{\tr_{\oc} \varphi}{ B_{\oc}\,
  \tr_{\oc}\varphi }_{\partial\Omega}\,.
\end{align*}
Reasoning as in \eqref{xi-bound} and using \eqref{tb} we obtain, for
any $\varphi\in \mathcal{D}_{\oc}$ and $\xi\in S_{\mu_0,\rho}$,
\begin{align*}
  \|(H_{\oc}^{(m)}-\xi)\varphi\|^2_{\oc}&\geqslant 
\tfrac{1}{2}
    \|H_{\oc}^{(m)}\varphi\|_{\oc}^2+(\mu_0^2-2\rho^2) 
    \|\varphi\|_{\oc}^2\\
&= \tfrac{1}{2}\|T\varphi\|^2_{\oc}
  +\tfrac{m}{2}\|\tr_{\oc}\varphi\|^2_{\partial\Omega}+(
\tfrac{m^2}{2}+\mu_0^2-2\rho^2) 
    \|\varphi\|_{\oc}^2\\
&\geqslant
  \tfrac{m}{2}\|\tr_{\oc}\varphi\|^2_{\partial\Omega}+\tfrac{m^2}{4}   \|\varphi\|_{\oc}^2\,,
\end{align*}
where in the last step we used the hypothesis on $m$. We find the
bounds \eqref{eq:13.2} and \eqref{eq:13.3} after inserting
$\varphi=R_{\oc}^{(m)}(\xi) h$ in the last inequality.
\end{proof}
%
%
%
The next statement follows from 
\cite[Lemma~4]{Stockmeyerwugalter2016}.
\begin{proposition}[A priori lower bound \cite{Stockmeyerwugalter2016}]\label{sw16}
There exist constants $c<\infty$ and $m_0'>1$ such that for all
$\varphi\in H^1(\R^2,\C^2)$
\begin{align*}
  \|H_m\varphi\|^2\geqslant \|\nabla \varphi\|_{\Omega}^2+m \|P_-
  \tr_{\Omega}\varphi\|^2_{\partial\Omega}-c
  \,\|\tr_{\Omega}\varphi\|^2_{\partial\Omega},\quad m>m_0'\,.
\end{align*}
\end{proposition}
\begin{remark}
This result is stated in \cite[Lemma~4]{Stockmeyerwugalter2016} for
a bounded and regular domain $\Omega$. However, its proof extends
easily to unbounded
domains. It also extends to the case in which
different components of the boundary are at a certain minimal distance
from each other and when we can choose global tubular coordinates with
a globally bounded curvature.  Hypothesis~\ref{hyp0} on the boundary
guarantees just that. 

For completenes let us quickly review the argument: In view of
Lemma~\ref{part-int} one can check that, for any
$\varphi\in H^1(\R^2,\C^2)$,
\begin{align*}
   \|H_m\varphi\|^2&=\|T\varphi\|^2+m^2\|\varphi\|^2_{\oc}+2m{\rm Re}
                     \sps{T\varphi}{\sigma_3\varphi}_{\oc}\\
&=\|\nabla \varphi\|^2+m^2\|\varphi\|^2_{\oc}+m(\|P_- \tr\,\varphi
  \|^2_{\partial\Omega}-\|P_+ \tr\,\varphi \|^2_{\partial\Omega})\,.
\end{align*}
Thus, it suffices to show that, for $m$ large enough, 
\begin{align}
  \label{eq:222}
  \|\nabla \varphi\|^2_{\oc}+m^2\|\varphi\|^2_{\oc}\ge (m-c) \| \tr\,\varphi \|^2_{\partial\Omega}\,.
\end{align}
Recall the definition of $\delta_0$ given in
Hypothesis~\ref{hyp0}. The next step in the proof given in 
\cite{Stockmeyerwugalter2016} uses a smooth partition of the unity $u,v$ on the set 
$Q_{\delta_0}=\{x\in \oc\,:\, {\rm dist}(x,\partial\Omega)<\delta_0\}$
such that $u^2+v^2=1$ on $\oc$ and $u$ is supported in  $Q_{\delta_0}$
with $u(x)=1$ when $ {\rm dist}(x,\partial\Omega)<\delta_0/2$. Using
the IMS localisation formula one gets, for some $c>0$, 
\begin{align*}
   \|\nabla \varphi\|^2_{\oc}+m^2\|\varphi\|^2_{\oc}\ge
 \|\nabla (u\varphi)\|^2+(m^2-c^2)\|u\varphi\|^2_{\oc}\,.
\end{align*}
The latter expression, locates the problem only on $Q_{\delta_0}$. One
can then obtain \eqref{eq:222} by expressing the integrals in tubular
coordinates  and using \cite[Lemma~4]{Stockmeyerwugalter2016}.
\end{remark}
\begin{corollary}\label{r-lemma3}
There exist a constant $c'<\infty$ such that for all $\xi\in
S_{\mu_0,\rho}$, $f\in L^2(\R^2,\C^2)$,  and all $m>m_0'$
\begin{align}
  \label{eq:13.4}
  &\|P_- \tr_{\Omega} R_{m}(\xi) f\|_{\partial\Omega} \leqslant \frac{c'}{\sqrt{m}} \|f\|\\\label{eq:13.5}
&\|\tr_{\Omega} R_{m}(\xi) f\|_{\partial\Omega}\leqslant c'\,\|f\|\,.
\end{align}
\end{corollary}
\begin{proof}
  We follow the same strategy as in the proof of
  Lemma~\ref{r-lemma1}. By Proposition~\ref{sw16} and \eqref{ti} we
  have, for $\varepsilon^{-1}\geqslant 1+c $,
  \begin{align*}
    \|H_m\varphi\|^2\geqslant
    \|\tr_{\Omega}\varphi\|^2_{\partial\Omega}+
m \|P_- \tr_{\Omega}\varphi\|^2_{\partial\Omega}
    -\tfrac{C_\varepsilon}{\varepsilon}\|\varphi\|^2_{\Omega} \,.
  \end{align*}
Therefore, $2\|  (H_m-\xi)\varphi\|^2\geqslant
\|\tr_{\Omega}\varphi\|^2_{\partial\Omega}
+
m \|P_- \tr_{\Omega}\varphi\|^2_{\partial\Omega}-\tilde{C}_\varepsilon
\|\varphi\|^2$, where
$\tilde{C}_\varepsilon:=|2\mu_0^2-4\rho^2-C_\varepsilon/\varepsilon|$.
Thus, with the substitution $\varphi=R_m(\xi) f$, we get that
\begin{align*}
  \|\tr_{\Omega} R_m(\xi) f\|^2_{\partial\Omega}+m \|P_-
  \tr_{\Omega}R_m(\xi) f\|^2_{\partial\Omega}
\leqslant (2+\tilde{C}_\varepsilon/\mu_0^2) \|f\|^2,
\end{align*}
from which follows the claim.
\end{proof}
%
%
%
\begin{proof}[Proof of Lemma~\ref{energysquared}]
  Let $\varphi\in C^1(\overline{\og},\C^2)\cap\mathcal{D}(\og)$ be a continuous
  differentiable spinor. Then, a direct computation yields
\begin{align*}
   \|T
  \varphi\|_{\og}^2=\|\nabla\varphi\|^2_{\og}+i\int_{\partial\Omega}
  (\varphi(\bx), \sigma_3 ({\bf t}_{\og}\cdot \nabla)\,\varphi(\bx) ) \,d\omega_\bx\,.
\end{align*}
Assume further that $\varphi\in \mathcal{D}_{\og}$ and write
$\partial_{\bf t} := ({\bf t}\cdot \nabla)$ (we skip the index reference to
the set $\og$). Then, according to
Remark \ref{B-bc}, we find that on $\partial \Omega$ 
\begin{align*}
(\varphi,\sigma_3 \partial_{\bf t} \varphi)=  (\varphi,\sigma_3
  (\partial_{\bf t} B) \varphi)+ (\varphi,\sigma_3
   B\, \partial_{\bf t}  \varphi)=(\varphi,\sigma_3
  (\partial_{\bf t} B) \varphi)- (\varphi,\sigma_3
   \partial_{\bf t}  \varphi),
\end{align*}
where in the last inequality we use that $\{B,\sigma_3\}=0$. This
implies that $(\varphi,\sigma_3
   \partial_{\bf t}  \varphi)=(\varphi,\sigma_3
  (\partial_{\bf t} B) \varphi)/2$. Moreover, using \eqref{curvature},
  we get on $\partial \Omega$ that
  \begin{align*}
    i\sigma_3 \partial_{\bf t} B =\partial_{\bf t} \pauli\cdot\normal=
    \curv B\,.
  \end{align*}
Thus, for $\varphi \in C^1(\overline{\og},\C^2)\cap
\mathcal{D}_{\og}$,
\begin{align*}
    \|T
  \varphi\|_{\og}^2=\|\nabla\varphi\|^2_{\og}+\frac{1}{2}\int_{\partial\Omega}
  (\varphi(\bx),\curv \,\varphi(\bx) ) \,d\omega_\bx\,.
\end{align*}
Using a density argument we can extend the above expression to
$\mathcal{D}_{\og}$ obtaining \eqref{eq:14}.
\end{proof}

\section{Proof of Proposition~\ref{prop:sa} and Theorem~\ref{main-thm-2}}\label{S4}

\begin{proof}[Proof of Proposition~\ref{prop:sa}]
We first prove self-adjointness of $\hov$. Lemma~\ref{energysquared} and trace inequalities \eqref{trace-ineq} give, for $\epsilon\in(0,1)$ and all $\varphi\in\dom_\Omega$, with $C_\epsilon$ given by \eqref{trace-ineq}
\begin{equation}\label{lem-sa3}
\begin{split}
 \|T \varphi\|^2_\Omega & = \|\nabla\varphi\|^2_\Omega  
 + \frac12 \sps{\tr_{\Omega}\varphi}{\curv_{\Omega} \, 
 \tr_{\Omega} \varphi}_{\partial\Omega} \\
 & \geqslant \left(1-\epsilon\frac{\| \kappa_\Omega \|_{\infty} }{2}\right)  
 \|\nabla \varphi\|^2 - C_\epsilon \frac{\| \kappa_\Omega \|_{\infty} }{2} 
 \|\varphi\|^2_\Omega \, .
\end{split}
\end{equation}
Moreover, from Hypothesis~\ref{hyp1}(i), we have, for $\varphi\in\dom_\Omega\subset H^1(\Omega, \C^2)$, 
\begin{equation}\label{lem-sa4}
\begin{split}
 \|V\varphi\|^2_\Omega \leqslant a \|\nabla\varphi\|^2_\Omega + b \|\varphi\|^2_\Omega \, .
\end{split}
\end{equation} 
Picking $\epsilon$ small enough, \eqref{lem-sa3}-\eqref{lem-sa4} imply that there exists $a'\in(0,1)$ and $b'$ such that for all $\varphi\in\dom_\Omega$  
\begin{equation}\label{eq:rel-bound-2}
 \|V \varphi\|^2_\Omega \leqslant a' \|T\varphi\|^2_\Omega + b' \|\varphi\|^2_\Omega \, ,
\end{equation}
which proves that $\hov$ is self-adjoint on $\dom_\Omega$. 

The proof of self-adjointness of $\hmv$ is a direct consequence of the next lemma.
\end{proof}
\begin{lemma}\label{lem:V1}
Assume Hypothesis~\ref{hyp1} is fulfilled. Then, for $m$ large enough, $V$ is relatively bounded with respect to $H_m$, with relative bound less than one. Namely there exists $\tilde{a}\in (0,1)$ and $\tilde{b}\geqslant 0$ such that for $m>m_0$ and $\varphi\in H^1(\R^2, \C^2)$,
\begin{equation}\label{hyp2}
 \| V \varphi\|^2 \leqslant \tilde{a} \| H_m\varphi \|^2 + \tilde{b} \| \varphi \|^2 .
\end{equation}
\end{lemma}
\begin{proof}
  From Proposition~\ref{sw16}, there exists $c>0$ and $m_0>0$ such
  that for $m>m_0$ and $\varphi\in H^1(\R^3,\C^2)$, $\|H_m\varphi\|^2
  \geqslant \|\nabla \varphi\|^2_\Omega + m\|P_- \tr_\Omega
  \varphi\|_{\partial\Omega}^2 - c
  \|\tr_\Omega\varphi\|^2_{\partial\Omega}$. Hence using trace
  inequalities \eqref{trace-ineq}, we obtain that for all $\epsilon\in
  (0,1)$, there exists $c_\epsilon>0$ such that for all $\varphi\in
  H^1(\R^3,\C^2)$
\begin{equation}\nonumber
 \| H_m \varphi\|^2 \geqslant (1-\epsilon)\|\nabla\varphi\|^2_\Omega - c_\epsilon\|\varphi\|_\Omega^2 .
\end{equation}
This implies, from Hypothesis~\ref{hyp1}(i), 
\begin{equation}\label{lem:sa-1}
 \|V\varphi\|_\Omega^2 \leqslant \frac{a}{1-\epsilon} \|H_m\varphi\|^2 + \left(\frac{a c_\epsilon}{1-\epsilon} + b\right)\|\varphi\|_\Omega^2 .
\end{equation}
Now using the Hypothesis~\ref{hyp1}(ii), implies
\begin{equation}\label{lem:sa-2}
 \|V\varphi\|^2 \leqslant \|V \varphi\|_{\Omega}^2 + \|V\|^2_{L^\infty(\oc)} \|\varphi\|^2.
\end{equation}
In view of estimates \eqref{lem:sa-1} and \eqref{lem:sa-2}, we conclude by picking $\epsilon$ 
small enough so that $\tilde{a}:=a/(1-\epsilon) <1$.
\end{proof}

We are now equipped to prove Theorem~\ref{main-thm-2}. 
\begin{proof}[Proof of Theorem~\ref{main-thm-2}]
As in the proof of Theorem~\ref{main-thm},
for fixed constants $\mu_0>0$ and $\rho\geqslant 0$, we consider the set $S_{\mu_0,\rho}$ defined by \eqref{set}.
Pick $\xi\in S_{\mu_0,\rho}$ and denote $\rmv=\rmv(\xi)$, $\rov=\rov(\xi)$, $R_\Omega = R_\Omega(\xi)$ and $R_m = R_m(\xi)$. 

From Hypothesis~\ref{hyp1}(i) and \eqref{eq:rel-bound-2}, we have, in view of \cite[Satz~9]{Weidmann1976}, for $\mu_0$ large enough
\begin{equation}\label{eq:lem4.2-1}
\| R_\Omega  V\|<1 \, .
\end{equation}
Similarly, with \eqref{hyp2}, we have for $\mu_0$ large enough
\begin{equation}
\| V R_m \| < 1 \, .
\end{equation}
Hence, using resolvent formulae for $\rmv$ and $\rov\oplus 0$, we obtain 
\begin{equation}\label{eq:lem4.2-2}
 \rmv - \rov\oplus 0 = (1- R_\Omega V \oplus 0)^{-1} (R_m - R_\Omega\oplus 0)(1+ V \rmv)\, .
\end{equation}
In view of Theorem~\ref{main-thm}, using  \eqref{eq:lem4.2-1}-\eqref{eq:lem4.2-2} and the fact that
\begin{equation*}
V\rmv  = V R_m(1 + V R_m)^{-1}\, ,
\end{equation*}
implies that there exists $\mu_0>0$,  $m_0>1$ and $C>0$ such that for all $m>m_0$,   
$$
 \sup_{\xi\in S_{\mu_0,\rho}} \|\rmv(\xi) - \rov(\xi)\oplus 0\| \leqslant \frac{C}{\sqrt{m}}\, .
$$
The result then follows from Lemma~\ref{nc} that remains true replacing  $H_m$ 
and $H_\Omega$ respectively by $\hmv$ and $\hov$ and replacing in the statement Theorem~\ref{main-thm} by Theorem~\ref{main-thm-2}.
\end{proof}


 \appendix
 \section{Sufficient condition}
For the next lemma consider the setting of Theorem \ref{main-thm}. It
gives a sufficient condition to control the supremum norm in \eqref{eq:7}.
\begin{lemma}\label{nc} 
  Let $K\subset \varrho(H_{\Omega})$ be a compact set. For constants
  $\mu_0\in(\max_{z\in K}|{\rm Im}(z)|,\infty)$ and
  $\rho \in (0,\max_{z\in K}|{\rm Re}(z)|)$ recall the definition of
  the set $S_{\mu_0,\rho}$ from \eqref{set}. Assume that for some $C'<\infty$ and
  $m_0'>1$ we have that
  \begin{align}
    \label{eq:11}
    \sup_{z\in S_{\mu_0,\rho}}\|R_m(z)-R_\Omega(z)\oplus 0\|\leqslant
    \frac{C'}{\sqrt{m}}, 
  \end{align}
for all $m>m_0'$. Then, the statement of Theorem \ref{main-thm} holds
for some $m_0>m_0'$.
\end{lemma}
\begin{proof}
  Assuming \eqref{eq:11} we will show Theorem \ref{main-thm}, separately,
  in the two compact regions
$K^+:=\{z\in K\,:\,{\rm Im} (z)\geqslant 0\}$ and 
$K^-:=\{z\in K\,:\,{\rm Im} (z)\leqslant 0\}.$
  We define the connection between the sets
  $K^+$  and $S_{\mu_0,\rho}$ through the map 
$$K^+\ni z\mapsto\xi_z:={\rm Re} (z)+i\mu_0\in S_{\mu_0,\rho}\,.$$

We first show that $K^+\subset \varrho(H_m)$ for sufficiently large
$m>1$. For $\varphi\in H^1(\R^2,\C^2)$ and $z\in K^+$ we have
\begin{align}\label{a1}
  (H_m-z)\varphi&=(H_m-\xi_{z})\varphi+(\xi_{z}-z) R_m(\xi_{z}) (H_m-\xi_{z})\varphi
\end{align}
(above and along this proof we simply write $R_{\Omega}$ for
$R_{\Omega}\oplus 0$).

It is easy to verify that  the function
\begin{align}
  \label{eq:10}
  \|(\xi_{z}-z) R_{\Omega}(\xi_{z})\|=\frac{\mu_0 -{{\rm Im}(z)}
  }{{\rm dist} (\xi_{z}, {\rm spec}(H_{\Omega}))}\,, \quad z\in K^+,
\end{align}
is continuous and strictly smaller than 1 (here
${\rm dist}(\cdot,\cdot)$ denotes the distance from a point to a
set). Thus, since $K^+$ is compact, there is a $\delta\in(0,1)$ such that
\begin{align}
  \label{eq:12}
   \sup_{z\in K^+}\|(z-\xi_{z}) R_{\Omega}(\xi_{z})\|=1-\delta.
\end{align}
Writing $(\xi_{z}-z) R_m(\xi_{z}) =(\xi_{z}-z) R_{\Omega}(\xi_{z}) +(\xi_{z}-z) (R_m (\xi_{z})-R_{\Omega} (\xi_{z}))$
we see that, using   \eqref{eq:12} and \eqref{eq:11}, we may pick $m$ so large that 
\begin{align}
  \label{eq:13}
  \|(\xi_{z}-z) R_m(\xi_{z})\| \leqslant
  1 - \delta + \frac{2\mu_0C'}{\sqrt{m}}\leqslant 1-\delta/2,
\end{align}
uniformly for $z\in K^+$. This, together with \eqref{a1} , implies that 
\begin{align}\label{b1}
  \| (H_m-z)\varphi\|\geqslant \delta/2  \|
  (H_m-\xi_z)\varphi\|\geqslant \mu_0 \delta/2\|\varphi\|.
\end{align}
Therefore, $z\in \varrho(H_m)$ for $m$ sufficiently large. One can
show analogously that $K^{-}\subset \varrho(H_m)$ for $m$ large
enough, but this time using the connection $K^{-} \ni z\mapsto {\rm Re}
(z)-i\mu_0\in S_{\mu_0,\rho}\, $.

Next, we show that the limit in \eqref{eq:7} holds for $K$ replaced by
$K^+$. Let $m$ be so large that \eqref{eq:13} holds. Then, we have
that $K^+\subset \varrho(H_m)$ and, for any $z\in K^+$, we may use the
first resolvent identity to get
\begin{align*}
  R_{\Omega}(z)-R_m(z)&=R_{\Omega}(\xi_z) (1+(\xi_z-z)
  R_{\Omega}(\xi_z))^{-1} -R_{m}(\xi_z) (1+(\xi_z-z)
  R_{m}(\xi_z))^{-1}\\
&=:R_\Omega (\xi_z) M_{\Omega}(z)-R_{m}(\xi_z) M_{{m}}(z).
\end{align*}
Notice that $ M_m(z)-M_{\Omega}(z)=(\xi_z-z) M_m(z)
(R_{\Omega}(\xi_z)-R_{m}(\xi_z))   M_{\Omega}(z)$. 
Thus, in view of \eqref{eq:12} and  \eqref{eq:13}, we get that
\begin{align*}
  \|R_{\Omega}(z)-R_m(z)\|&\leqslant \|R_m(\xi_z)-R_{\Omega}(\xi_z)\|
  \|M_m(z)\| +\|R_{\Omega}(\xi_z)\| \|M_m(z)-M_{\Omega}(z)\|\\
&\leqslant \frac{2C'}{\sqrt{m} \delta}+\frac{4C'}{\sqrt{m} \delta^2 }\,.
\end{align*}
This proves  the claim for $K^+$. For $K^-$ one proceeds analogously 
with the aforementioned connection. This finishes the proof of the lemma.
\end{proof}
\bigskip

\noindent
{\cb {\bf Acknowledgments.}
It is a pleasure to thank the REB program of CIRM for giving us the
opportunity to generate this research. 
E.S has been partially funded by Fondecyt (Chile)
project \# 114--1008 and project \# 118--0355. L. L.T. was partially supported by ANR DYRAQ ANR-17-CE40-0016-01.}

\end{document}